\documentclass[12pt,a4paper]{article}
\usepackage{epsf, graphicx}
\usepackage{latexsym,amsfonts,amsbsy,amssymb}
\usepackage{amsmath,amsthm}
\usepackage{cleveref}
\usepackage{indentfirst}
\usepackage{bm}
\makeatletter

\@addtoreset{equation}{section} \makeatother
\newtheorem{Theorem}{Theorem}[section]
\newtheorem{Lemma}{Lemma}[section]
\newtheorem{Corollary}{Corollary}[section]
\newtheorem{Remark}{Remark}[section]
\newtheorem{Definition}{Definition}[section]
\newtheorem{Proposition}{Proposition}[section]
\newtheorem{Example}{Example}[section]
\makeatletter \@addtoreset{equation}{section} \makeatother
\begin{document}
\title{New Constructions of MDS Twisted Reed-Solomon Codes and LCD MDS Codes }

\author{Hongwei Liu,~Shengwei Liu}
\date{School of Mathematics and Statistics,
 Central China Normal University, Wuhan 430079, China}
\maketitle
\insert\footins{\small{\it Email addresses}: ~
 hwliu@mail.ccnu.edu.cn (Hongwei Liu); ~
 shengweiliu@mails.ccnu.edu.cn (Shengwei Liu).\\
 }
\section*{Abstract}
 \addcontentsline{toc}{section}{\protect Abstract} %加入目录
 \setcounter{equation}{0} %公式编号从1开始
  Maximum distance separable (MDS) codes are optimal where the minimum distance cannot be improved for a given length and code size. Twisted Reed-Solomon codes over finite fields were introduced in 2017, which are generalization of Reed-Solomon codes. Twisted Reed-Solomon codes can be applied in cryptography which prefer the codes with large minimum distance. MDS codes can be constructed from twisted Reed-Solomon codes, and most of them are not equivalent to Reed-Solomon codes. In this paper, we first generalize twisted Reed-Solomon codes to generalized twisted Reed-Solomon codes, then we give some new explicit constructions of MDS (generalized) twisted Reed-Solomon codes. In some cases, our constructions can get  MDS codes with the length longer than the constructions of previous works. Linear complementary dual (LCD) codes are linear codes that intersect with their duals trivially. LCD codes can be applied in cryptography. This application of LCD codes renewed the interest in the construction of LCD codes having a large minimum distance. We also provide new constructions of LCD MDS codes from generalized twisted Reed-Solomon codes.

\medskip
\noindent{\large\bf Keywords: }\medskip  Twisted Reed-Solomon codes, MDS codes, Linear complementary dual codes

\noindent{\bf2010 Mathematics Subject Classification}: 94B05, 94B65.

\section{Introduction}
  A linear code $C$ of length $n$, dimension $k$ and minimum Hamming distance $d$ over the finite field $\mathbb{F}_{q}$ is called an $[\,n, k, d\,]$ code. If the parameters of the code $C$ reach the Singleton bound $d=n-k+1$, then $ C$ is called a {\it maximum distance separable (MDS) code} \cite{R. Roth}. The most prominent MDS codes are {\it generalized Reed-Solomon (GRS) codes} \cite{R. Roth}. Recently, {\it twisted Reed-Solomon (TRS) codes} were firstly introduced in \cite{P. Beelen} by Beelen, Puchinger, and Rosenkilde n\'{e} Nielsen as a generalization of Reed-Solomon codes. In general, twisted Reed-Solomon codes do not necessarily lead to MDS codes. The authors in \cite{P. Beelen} gave two constructions of MDS twisted Reed-Solomon codes and showed that for $q\geq 11$,  this class of TRS codes contains non-GRS MDS codes. The authors also showed that twisted Reed-Solomon codes could be well decoded. Afterwards, they generalized the above constructions by adding extra monomial (twise), and obtained a construction of MDS twisted Reed-Solomon codes and showed that twisted Reed-Solomon codes could be applied to code-based cryptography, in other words, they resist some existing structural attacks for Reed-Solomon-like codes \cite{M. Bossert}. In \cite{Julien}, Lavauzelle and Renner presented an efficient key-recovery attack by twisted Reed-Solomon codes that was used in the McEliece cryptosystem. Furthermore, those applications prefer twisted Reed-Solomon codes with a large minimum distance. Therefore, it is meaningful to give more explicit constructions of MDS twisted Reed-Solomon codes. In this work, we give new constructions of MDS twisted Reed-Solomon codes, and in some cases, our constructions can get some MDS codes with the length longer than the constructions in \cite{P. Beelen}.

  The dual code of a linear code $C$ in $\mathbb{F}_{q}^{n}$ is denoted by $ C^{\bot}$. If $ C\subseteq  C^{\bot}$, then $ C$ is called a {\it self orthogonal code}. If $ C \bigcap C^{\bot} = \{0\}$, $ C$ is called a {\it linear complementary dual(LCD) code}. LCD codes were introduced by Massey \cite{J. L. Massey}. Then, many authors studied them (e.g., see \cite{H. Q. Dinh}-\cite{X. Hou}, \cite{C. Ding}-\cite{J. L. Massey} and \cite{Y. Wu}-\cite{Z. Zhou}). LCD codes were widely applied in coding theory and cryptography, especially in designing decoding algorithm. Carlet and Guilley showed that LCD codes are important in information protection and armoring implementations against side-channel attacks and fault non-invasive attacks in \cite{Guilley}. This application of LCD codes prefer the codes with a large minimum distance. Thus, it is significant to construct LCD MDS codes in theory and practice. In \cite{J. Qian}, Qian and Zhang constructed LCD MDS codes from constacyclic codes. In \cite{Jin L.2}, Jin constructed several classes of LCD MDS codes through generalized Reed-Solomon codes. In \cite{P. Beelen2}, Beelen and Jin gave an explicit construction of several classes of LCD MDS codes, using tools from algebraic function fields. In \cite{B. Chen}, Chen and Liu gave a different approach to obtain new LCD MDS codes from generalized Reed-Solomon codes which extended the results by Jin. In \cite{Xueying Shi}, Shi, Yue and Yang constructed some new LCD MDS codes from generalized Reed-Solomon codes. In \cite{Carlet C.}, Carlet et al. showed that linear codes over $\mathbb{F}_{q}$ are equivalent to LCD codes for $q>3$. In order to construct LCD MDS codes, we slightly generalize twisted Reed-Solomon codes and provide new constructions of LCD MDS codes from generalized twisted Reed-Solomon codes.

  This paper is organized as follows. Section $2$ gives the preliminaries. In Section $3$, we give new explicit constructions of MDS twisted Reed-Solomon codes. In Section $4$, LCD MDS codes are constructed from twisted generalized Reed-Solomon codes. Section~5 concludes our work.

\section{Preliminaries}
Let $\mathbb{F}_{q}$ be the finite field of $q$ elements, where $q$ is a prime power. Let $\mathbb{F}_{q}^{*}=\mathbb{F}_{q}\backslash\{0\}$ be the multiplicative group of $\mathbb{F}_q$. The $(i,j)$th entry of a matrix $A$ over $\mathbb{F}_q$ is denoted by $A_{ij}$. The transpose of $A$ is denoted by $A^{T}$. The size of a finite set $S$ is denoted by $|S|$. The set $S^{n}$ for a set $S$ is defined to be $\{(s_{1},\dots,s_{n}):s_{i}\in S,~i=1,\dots,n\}$. A linear code $ C$ of length $n$, dimension $k$ and minimum distance $d$ over $\mathbb{F}_{q}$ is called an $[\,n, k, d\,]$ code. If the parameters of the code $ C$ reach the Singleton bound: $d=n-k+1$, then $C$ is called a maximum distance separable (MDS) code.

The dual code of a linear code $C$ in $\mathbb{F}_{q}^{n}$ is denoted by $ C^{\bot}$. If $ C\subseteq  C^{\bot}$, then $C$ is called a self orthogonal code. If $ C \bigcap C^{\bot} = \{0\}$, then $ C$ is called a linear complementary dual (LCD) code. In this paper, the inner product between two vectors over $\mathbb{F}_q$ is always the Euclidean inner product.

Let $\mathbb{F}_{q}[x]$ be the polynomial ring over $\mathbb{F}_{q}$.
Let $\bm\alpha=(\alpha_{1},\alpha_{2},...,\alpha_{n})$ and $\bm v=(v_1,v_2\cdots, v_n)$ be two vectors of length $n$ over $\mathbb{F}_q$. We define the evaluation map related to $\bm\alpha$ and $\bm v$ as follows:

$$
ev_{\bm\alpha, \bm v}:\mathbb{F}_{q}[x]\rightarrow\mathbb{F}_{q}^{n},~f(x) \mapsto (v_1f(\alpha_{1}),v_2f(\alpha_{2}),...,v_nf(\alpha_{n})).
$$
\begin{Definition}\label{def-1}\cite{R. Roth}
Let $\alpha_{1},\dots,\alpha_{n}\in\mathbb{F}_{q}\bigcup\{\infty\}$ be distinct elements, $k<n$, and $v_{1},\dots,v_{n}\in\mathbb{F}_{q}^{*}$. The corresponding generalized Reed-Solomon (GRS) code is defined by
$$
GRS_{n,k} := \{(v_{1}f(\alpha_{1}),\dots,v_{n}f(\alpha_{n})): f\in \mathbb{F}_{q}[x], \deg f<k\}.
$$
In this setting, for a polynomial $f(x)$ of degree $\deg f(x)<k$, the quantity $f(\infty)$ is defined as the coefficient of $x^{k-1}$ in the polynomial $f(x)$. In case $v_{i}=1$ for all $i$, the code is called a Reed-Solomon (RS) code.
\end{Definition}

Let $1\leq\ell\leq k$, and $n\geq k$ be positive integers. Let $\bm h=(h_{1},h_{2},...,h_{\ell})\in\{0,1,...,k-1\}^{\ell}$ be a vector such that $h_{i}$ are all distinct and $h_{1}<h_{2}<\dots<h_{\ell}$, and let $\bm t=(t_{1},t_{2},...,t_{\ell})\in\{1,...,n-k\}^{\ell}$ such that all $t_{j}$ are distinct. Let $\bm \eta=(\eta_{1},\eta_{2},...,\eta_{\ell})\in\{\mathbb{F}_{q}\backslash0\}^{\ell}$. The set of $[\bm t,\bm h,\bm \eta]$-{\it twisted polynomials} is defined as

$$
\mathcal{P}_{k,n}[\bm t,\bm h,\bm\eta]
:=\mathbf{\{}\sum_{i=0}^{k-1}f_{i}x^{i}+\sum_{j=1}^{\ell}\eta_{j}f_{h_{j}}x^{k-1+t_{j}}:f_{i}\in\mathbb{F}_{q}\mathbf{\}}\subseteq\mathbb{F}_{q}[x].
$$
\begin{Definition}\label{def-2}
Let the entries of $\bm\alpha=(\alpha_{1},\alpha_{2},...,\alpha_{n})\in\mathbb{F}_{q}^{n}$ be pairwise distinct, $\bm v=(v_{1},v_{2},...,v_{n})\in (\mathbb{F}_{q}\setminus\{0\})^{n}$ and fix $1\leq k\leq n$. Let $\mathbf{t},\mathbf{h},\bm\eta$ and $\mathcal{P}_{k,n}[\mathbf{t},\mathbf{h},\bm\eta]$ be defined as above. The generalized twisted Reed-Solomon(GTRS)code of length n, dimension k and locators $\bm\alpha$ is defined by

$$
GTRS_{k,n}[\bm\alpha,\bm t,\bm h,\bm\eta,\bm v]:=\{ev_{\bm\alpha, \bm v}(f): f\in\mathcal{P}_{k,n}[\mathbf{t},\mathbf{h},\bm\eta]\}.
$$
\end{Definition}
By the definition of GTRS codes, a generator matrix of $GTRS_{k,n}[\bm\alpha,\bm t,\bm h,\bm\eta,\bm v]$ is given by
$$
G=\left[
\begin{matrix}
v_{1}&v_{2}&\dots&v_{n}\\
v_{1}\alpha_{1}&v_{2}\alpha_{2}&\dots&v_{n}\alpha_{n}\\
\vdots&\vdots&&\vdots\\
v_{1}\alpha_{1}^{h_{1}-1}&v_{2}\alpha_{2}^{h_{1}-1}&\dots&v_{n}\alpha_{n}^{h_{1}-1}\\
v_{1}(\alpha_{1}^{h_{1}}+\eta_{1}\alpha_{1}^{k-1+t_{1}})&v_{2}(\alpha_{2}^{h_{1}}+\eta_{1}\alpha_{2}^{k-1+t_{1}})&\dots&v_{n}(\alpha_{n}^{h_{1}}+\eta_{1}\alpha_{n}^{k-1+t_{1}})\\
v_{1}\alpha_{1}^{h_{1}+1}&v_{2}\alpha_{2}^{h_{1}+1}&\dots&v_{n}\alpha_{n}^{h_{1}+1}\\
\vdots&\vdots&&\vdots\\
v_{1}\alpha_{1}^{h_{\ell}-1}&v_{2}\alpha_{2}^{h_{\ell}-1}&\dots&v_{n}\alpha_{n}^{h_{\ell}-1}\\
v_{1}(\alpha_{1}^{h_{\ell}}+\eta_{\ell}\alpha_{1}^{k-1+t_{\ell}})&v_{2}(\alpha_{2}^{h_{\ell}}+\eta_{\ell}\alpha_{2}^{k-1+t_{\ell}})&\dots&v_{n}(\alpha_{n}^{h_{\ell}}+\eta_{\ell}\alpha_{n}^{k-1+t_{\ell}})\\
v_{1}\alpha_{1}^{h_{\ell}+1}&v_{2}\alpha_{2}^{h_{\ell}+1}&\dots&v_{n}\alpha_{n}^{h_{\ell}+1}\\
\vdots&\vdots&&\vdots\\
v_{1}\alpha_{1}^{k-1}&v_{2}\alpha_{2}^{k-1}&\dots&v_{n}\alpha_{n}^{k-1}
\end{matrix}
\right]_{k\times n}.
$$

If the vector $\bm v=\bm 1$, the all $1$'s vector, i.e., $v_1=v_2=\cdots=v_n=1$, then the GTRS code is called the {\it twisted Reed-Solomon code}, which was first introduced in \cite{M. Bossert} as a generalization of the RS code \cite{P. Beelen} \cite{M. Bossert}. For short, we denote it by
$$
TRS_{k,n}[\bm\alpha,\mathbf{t},\mathbf{h},\bm\eta]:=GTRS_{k,n}[\bm\alpha,\bm t,\bm h,\bm\eta,\bm 1]=\{ev_{\bm\alpha, \bm 1}(f):f\in\mathcal{P}_{k,n}[\mathbf{t},\mathbf{h},\bm\eta]\}.
$$

Obviously, a generator matrix of $TRS_{k,n}[\bm\alpha,\bm t,\bm h,\bm\eta]$ is given by

$$
G_{\bm\alpha,\bm t,\bm h,\bm\eta}=\left[
\begin{matrix}
1&1&\dots&1\\
\alpha_{1}&\alpha_{2}&\dots&\alpha_{n}\\
\vdots&\vdots&&\vdots\\
\alpha_{1}^{h_{1}-1}&\alpha_{2}^{h_{1}-1}&\dots&\alpha_{n}^{h_{1}-1}\\
\alpha_{1}^{h_{1}}+\eta_{1}\alpha_{1}^{k-1+t_{1}}&\alpha_{2}^{h_{1}}+\eta_{1}\alpha_{2}^{k-1+t_{1}}&\dots&\alpha_{n}^{h_{1}}+\eta_{1}\alpha_{n}^{k-1+t_{1}}\\
\alpha_{1}^{h_{1}+1}&\alpha_{2}^{h_{1}+1}&\dots&\alpha_{n}^{h_{1}+1}\\
\vdots&\vdots&&\vdots\\
\alpha_{1}^{h_{\ell}-1}&\alpha_{2}^{h_{\ell}-1}&\dots&\alpha_{n}^{h_{\ell}-1}\\
\alpha_{1}^{h_{\ell}}+\eta_{\ell}\alpha_{1}^{k-1+t_{\ell}}&\alpha_{2}^{h_{\ell}}+\eta_{\ell}\alpha_{2}^{k-1+t_{\ell}}&\dots&\alpha_{n}^{h_{\ell}}+\eta_{\ell}\alpha_{n}^{k-1+t_{\ell}}\\
\alpha_{1}^{h_{\ell}+1}&\alpha_{2}^{h_{\ell}+1}&\dots&\alpha_{n}^{h_{\ell}+1}\\
\vdots&\vdots&&\vdots\\
\alpha_{1}^{k-1}&\alpha_{2}^{k-1}&\dots&\alpha_{n}^{k-1}
\end{matrix}
\right]_{k\times n}.
$$

\begin{Definition}\label{def-3}
Let $C_{1}$, $C_{2}$ be $\mathbb{F}_{q}$-linear $[~n,k~]$ codes. We say that $C_{1}$ and $C_{2}$ are {\it equivalent} if there is a permutation $\pi\in S_{n}$ and $\bm v=(v_{1},\dots,v_{n})\in(\mathbb{F}_{q}^{*})^{n}$ such that $C_{2}=\varphi_{\pi,\bm v}(C_{1})$ where $\varphi_{\pi,\bm v}$ is
$$
\varphi_{\pi,\bm v}:\mathbb{F}_{q}^{n}\rightarrow\mathbb{F}_{q}^{n},~(c_{1},\dots,c_{n})\mapsto(v_{1}c_{\pi(1)},\dots,v_{n}c_{\pi(n)}).
$$
\end{Definition}
It is easy to see that $C_{1}$ and $C_{2}$ have same parameters, and a code is a GRS code if and only if it is equivalent to an RS code.

\section{New explicit constructions of MDS TRS codes}
We know that the generalized twisted Reed-Solomon code $GTRS_{k,n}[\bm\alpha,\bm t,\bm h,\bm\eta,\bm v]$ is equivalent to the twisted Reed-Solomon code $TRS_{k,n}[\bm\alpha,\bm t,\bm h,\bm \eta]$. In this section, we give some explicit constructions of MDS TRS codes. And all results in this section can be generalized to GTRS codes.

 Let $K$ be a group, $H$ be a subgroup of $K$ and $a\in K$. Then the left (right) coset of $a$ with respect to $H$ is the set $aH=\{ah : h\in H\}$ ($Ha=\{ha : h\in H\}$). The index of $H$ in $K$, denoted $[K : H]$, is the number of left (right) cosets of $H$ in $K$. If $aH=Ha$ for all $a\in K$, then we say that $H$ is a normal subgroup of $K$. If $K$ is abelian, then all subgroup of $K$ is normal subgroup. The quotient group $K/H$ is the set of all left cosets $aH$ denoted by $\overline{a}$, with $a\in K$, under the operation $(aH)(bH)=abH$. For a group $K$, $|K|$ is called the order of $K$.

In~\cite{P. Beelen}, the authors have studied the cases of $\ell=1$ and $\bm v=\bm 1$. They given two explicit constructions of MDS TRS codes with the case $(t,h)=(1,0)$ which they called $(\ast)$-twisted codes and the case $(t,h)=(1,k-1)$ which they called $(+)$-twisted codes that contain many non-GRS codes. In this section, we generalize their constructions, and give two ways to increase the code length correspond to the cases $(t,h)=(1,0)$ and $(t,h)=(1,k-1)$ respectively. With notation in Section $2$, we present three useful results, which have been shown in \cite{P. Beelen}. Throughout this section, let $\ell=1$.

\begin{Lemma}\label{lem-1}\cite{P. Beelen}
Let $(t,h)=(1,0)$ and let $k<n$, $\alpha_{1},\alpha_{2},...,\alpha_{n}\in\mathbb{F}_{q}$ distinct and $\eta\in\mathbb{F}_{q}$. Then the twisted code $TRS_{k,n}[\bm\alpha,1,0,\eta]$ is MDS if and only if
$$
%\begin{eqnarray*}
\eta(-1)^{k}\prod_{i\in \mathcal{I}}\alpha_{i}\neq 1,~\forall~\mathcal{I}\subseteq \{1,...,n\}~\text s.t.~|\mathcal{I}|=k.
%\end{eqnarray*}
$$
\end{Lemma}
\begin{Lemma}\label{lem-2}\cite{P. Beelen}
Let $k<n$, $(t,h)=(1,k-1)$ and let $\alpha_{1},\alpha_{2},...,\alpha_{n}\in\mathbb{F}_{q}$ distinct and $\eta\in\mathbb{F}_{q}$. Then the twisted code $TRS_{k,n}[\bm\alpha,1,k-1,\eta]$ is MDS if and only if
$$
\eta\sum_{i\in \mathcal{I}}\alpha_{i}\neq -1,~\forall~\mathcal{I}\subseteq \{1,...,n\}~\text s.t.~|\mathcal{I}|=k.
$$
\end{Lemma}
\begin{Lemma}\label{lem-3}\cite{P. Beelen}
A linear code with generator matrix $G=[I:A]$ is a GRS code if and only if

(A) All entries of $A$ are nonzero.

(B) All $2\times2$ minors of $\widetilde{A}$ are nonzero, and

(C) All $3\times3$ minors of $\widetilde{A}$ are zero.

where $\widetilde{A}$ is given by $\widetilde{A}_{ij}=A^{-1}_{ij}$.
\end{Lemma}
From Lemma~\ref{lem-3}, the authors in \cite{P. Beelen} got the following result by calculating the minors.
\begin{Theorem}\label{theorem-1}\cite{P. Beelen}
Let $\alpha_{1},\alpha_{2},...,\alpha_{n}\in\mathbb{F}_{q}$ distinct and $2<k<n-2$. Furthermore, let $\mathcal{H}\subseteq \mathbb{F}_{q}$ satisfy that the twisted code $TRS_{k,n}[\bm\alpha,t,h,\eta]$ is MDS for all $\eta\in\mathcal{H}$. Then there are at most 6 choices of $\eta\in\mathcal{H}$ such that $TRS_{k,n}[\bm\alpha,t,h,\eta]$ is equivalent to an RS code.
\end{Theorem}
Let $G$ be a proper subgroup of $(\mathbb{F}_{q}^{*},\cdot)$, then $G$ is a cyclic group. Let $\mathbb{F}_q^*/G=\{\overline{a_{1}},\overline{a_{2}},...,\overline{a_{s}}\}$ be the quotient group with $\{a_{1},a_{2},...,a_{s}\}$ being a representation elements of cosets. For any $a\in \mathbb{F}_q^*$, let  $\overline{a}=aG:=\{ag:g\in G\}$ be a coset. Then we have the following result.
\begin{Theorem}\label{the-2}
Let $\overline{H}=\{\overline{a_{i_{1}}},\overline{a_{i_{2}}},\dots,\overline{a_{i_{r}}}\},r\geq1$ be a proper subgroup of $\mathbb{F}_{q}^{*}/G$ with order $r$, and let $(t,h)=(1,0), 1\leq k<n, \{\alpha_{1},\alpha_{2},...,\alpha_{n}\}$ be a subset of $\cup_{j=1}^{r}a_{i_{j}}G\bigcup\{0\}$. If $(-1)^{k}\eta\in \mathbb{F}_{q}^{*}\setminus\cup_{j=1}^{r}a_{i_{j}}G $, then $TRS_{k,n}[\bm\alpha,1,0,\eta]$ is an MDS code.
\end{Theorem}
\begin{proof} Suppose that the code $TRS_{k,n}[\bm\alpha,1,0,\eta]$ is not MDS, then by Lemma~\ref{lem-1}, there exists $\mathcal{I}\subseteq\{1,2,...,n\}$ such that $\eta(-1)^{k}\prod_{i\in \mathcal{I}}\alpha_{i}= 1$, which implies that
$$
\overline{\eta(-1)^{k}\prod_{i\in \mathcal{I}}\alpha_{i}}=\overline{\eta(-1)^{k}}\overline{\prod_{i\in \mathcal{I}}\alpha_{i}}= \overline{1}.
$$
If $\alpha_{j}=0$ for some $j$, then it is a contradiction. Thus $\{0\}$ is not contained in $\{\alpha_{1},\alpha_{2},...,\alpha_{n}\}$.
Since all $\alpha_{i}$ are contained in $\cup_{j=1}^{r}a_{i_{j}}G$, and $\overline{H}$ is a group, we have $\overline{\prod_{i\in\mathcal{I}}\alpha_{i}}\in \overline{H}$, implying that $\overline{(-1)^{k}\eta}\in \overline{H}$. This also leads to a contradiction. We finish the proof.
\end{proof}

Let the notions be as in Theorem~\ref{the-2}, then by Theorem~\ref{theorem-1}, we have the following result.
\begin{Corollary}\label{cor-1}
If $|\mathbb{F}_{q}^{*}\setminus\cup_{j=1}^{r}a_{i_{j}}G|>6$, then for any $n,k$ with $2<k<n-2$ and $n\leq |\cup_{j=1}^{r}a_{i_{j}}G|$ there exists a non-GRS MDS $TRS_{k,n}[\bm\alpha,1,0,\eta]$ code.
\end{Corollary}
\begin{Remark}
The construction of Theorem~\ref{the-2} is $(\ast)$-twisted codes in \cite{P. Beelen} when $G=\{1\}$.
\end{Remark}
\begin{Remark}
From the result, the code length can reach $r|G|+1$. If the quotient group $\mathbb{F}_{q}^{*}/G$ has proper non-trivial subgroup, then our construction increases the code length. However, it is easy to see that $r|G| | |\mathbb{F}_{q}^{*}|$ and $r|G| < |\mathbb{F}_{q}^{*}|$. Since $\mathbb{F}_{q}^{*}$ is a cyclic group, then for any divisor $c$ of $|\mathbb{F}_{q}^{*}|$ there exists a unique subgroup of order $c$. Then there exists a proper subgroup of $\mathbb{F}_{q}^{*}$ with order $r|G|$. So Theorem~\ref{the-2} can not increase the maximal length which $(\ast)$-twisted codes reached.
\end{Remark}
In coding theory, we always want to construct an MDS code with a longer code length. In \cite{P. Beelen}, when $(t,h)=(1,0)$ the code length of $(\ast)$-twisted codes can reach $t+1$ where $t$ is the maximal proper divisor of $q-1$. In the following, we give a construction which increases the maximal length that $(\ast)$-twisted codes reached for the case $(t,h)=(1,0)$ and $q$ is even.

Before giving the following result, we simply state some facts about integer $2^{m}-1$. We know that if $m$ is a composite number, then $2^{m}-1$ is also a composite number. Further more, $3$ is a divisor of $2^{m}-1$ if and only if $m$ is even. Thus, there are numerous $m$ such that $2^{m}-1$ is a composite number. However, when $m$ is a prime number, we can not ensure $2^{m}-1$ being a prime number. For example, $2^{29}-1=536870911=233\times1103\times2089$.
\begin{Theorem}\label{the-3}
Let $\mathbb{F}_{2^{m}}$ be a finite field such that $2^{m}-1$ is a composite number, and $p$ be the minimal prime divisor of $2^{m}-1$. Let $G$ be a subgroup of $(\mathbb{F}_{2^{m}}^{*},\cdot)$ with order $\frac{2^{m}-1}{p}$, and $\mathbb{F}_{2^{m}}^{*}/G=\{\overline{1},\overline{\gamma},...,\overline{\gamma}^{p-1}\}$ with $\overline{\gamma}$ the generator of $\mathbb{F}_{2^{m}}^{*}/G$. Let $a_{1},\dots, a_{p-2}\in \overline{\gamma}=\gamma G$ be distinct. Suppose $(t,h)=(1,0), 1\leq k<n, \{\alpha_{1},\alpha_{2},...,\alpha_{n}\}$ is a subset of $G\cup\{a_{1},\dots, a_{p-2},0\}$. If $\eta\in \overline{\gamma}=\gamma G$, then $TRS_{k,n}[\bm\alpha,1,0,\eta]$ is MDS.
\end{Theorem}
\begin{proof} Since $p$ is the minimal prime divisor of $2^{m}-1$, then we have $p\geq3$ and $p\leq|G|$. So $\{a_{1},\dots, a_{p-2}\}$ is not an empty set. If a $TRS_{k,n}[\bm\alpha,1,0,\eta]$ is not MDS, then Lemma~\ref{lem-1} implies that there exists $\mathcal{I}\subseteq\{1,2,...,n\}$ such that $\eta\prod_{i\in \mathcal{I}}\alpha_{i}= 1$, implying $$\overline{\eta\prod_{i\in \mathcal{I}}\alpha_{i}}=\overline{\eta}\overline{\prod_{i\in \mathcal{I}}\alpha_{i}}= \overline{1}.$$
If $\alpha_{j}=0$ for some $j$, then it is a contradiction. Thus $\{0\}$ is not contained in $\{\alpha_{1},\alpha_{2},...,\alpha_{n}\}$. Since the $\alpha_{i}$ are containted in $G\cup\{a_{1},\dots, a_{p-2}\}$, we have $\overline{\prod_{i\in\mathcal{I}}\alpha_{i}}= \overline{1}$ or $\overline{\gamma}$ or $\overline{\gamma}^{2}$ or$\dots$or $\overline{\gamma}^{p-2}$. Then $\overline{\eta}=\overline{1}$ or $\overline{\gamma}^{p-1}$ or $\overline{\gamma}^{p-2}$ or$\dots$or $\overline{\gamma}^{2}$. This gives a contradiction.
\end{proof}

Let the notions be as in Theorem~\ref{the-3}, then by Theorem~\ref{theorem-1} we have the following result.
\begin{Corollary}\label{cor-2}
If $|\gamma G|>6$, then for any $n,k$ with $2<k<n-2$ and $n\leq|G\cup\{a_{1},\dots, a_{p-2}\}|$ there exists a non-GRS MDS $TRS_{k,n}[\bm\alpha,1,0,\eta]$ code.
\end{Corollary}
\begin{Remark}
In the case $(t,h)=(1,0)$, from Theorem~\ref{the-3}, the code length can reach $\frac{2^{m}-1}{p} + p-1$ where $p$ is the minimal prime divisor of $2^{m}-1$. However, the length of $(\ast)$-twisted codes can only reach $\frac{2^{m}-1}{p} + 1$. But for a finite field $\mathbb{F}_{q}$ of odd $q$, the minimal prime divisor of $q-1$ is $2$, then we can not increase length by this way.
\end{Remark}
\begin{Example}
Let $q=2^{9}$, then $7$ is the minimal prime divisor of $q-1$. Let $\zeta$ be a generator of $(\mathbb{F}_{q}^{*},\cdot)$, and $G=\left\langle\zeta^{7}\right\rangle$. Let $\{a_{1},a_{2},a_{3},a_{4},a_{5}\}\subseteq\zeta G$ and $\{\alpha_{1},\dots,\alpha_{79}\}=G\cup\{a_{1},a_{2},a_{3},a_{4},a_{5},0\}$, then for any $1\leq k<79$ and $\eta\in\zeta G$, $TRS_{k,79}[\bm\alpha,1,0,\eta]$ is MDS. However, by the method in \cite{P. Beelen}, the length only reach $74$. Furthermore, there exist many non-GRS codes.
\end{Example}
Now, we give other constructions utilizing the group $(\mathbb{F}_{q}, +)$. Firstly, we simply state the construction of the abelian group $(\mathbb{F}_{q}, +)$. Let $q=p^{m}$ with $p$ a prime number $m\geq1$, then $\mathbb{F}_{q}$ is a $\mathbb{Z}_{p}$ linear space of dimension $m$. Suppose $\{\mu_{1},\dots,\mu_{m}\}$ is a $\mathbb{Z}_{p}$ base of $\mathbb{F}_{q}$, then $(\mathbb{F}_{q}, +)=\mathbb{Z}_{p}\mu_{1}\bigoplus\dots\bigoplus\mathbb{Z}_{p}\mu_{m}$ where $\bigoplus$ denotes the direct sum of groups. Every $\mathbb{Z}_{p}\mu_{j}$ is isomorphic to $\mathbb{Z}_{p}$. Since $\mathbb{Z}_{p}$ is a cyclic group of prime order, $\mathbb{Z}_{p}$ have only two subgroup $0$ and $\mathbb{Z}_{p}$. So, all subgroup of $(\mathbb{F}_{q}, +)$ are $\{0,(\mathbb{F}_{p}, +),(\mathbb{F}_{p^{2}}, +),\dots,(\mathbb{F}_{p^{m}}, +)\}$ in the isomorphism sense.

Let $q=p^{m}$, $V$ be a proper subgroup of $(\mathbb{F}_{q},+)$, $V$ is also a normal subgroup, since $(\mathbb{F}_{q},+)$ is abelian. For any $b\in \mathbb{F}_{q}$, let $\overline{b}=b+V:=\{b+v:v\in V\}$ be a coset. $\mathbb{F}_{q}/V = \{\overline{b_{1}},\overline{b_{2}},\dots,\overline{b_{s}}\}$ is the quotient group with $\{b_{1},b_{2},...,b_{s}\}$ being a representation elements of cosets. Then we have the following result.
\begin{Theorem}\label{the-4}
Let $\overline{L}=\{\overline{b_{i_{1}}},\overline{b_{i_{2}}},\dots,\overline{b_{i_{r}}}\},r\geq1$ be a proper subgroup of $\mathbb{F}_{q}/V$, and let $(t,h)=(1,k-1), 1\leq k<n, \{\alpha_{1},\alpha_{2},\dots,\alpha_{n}\}$ be a subset of $\cup_{j=1}^{r}(b_{i_{j}}+V)\bigcup\{\infty\}$. If $(-\eta)^{-1}\in \mathbb{F}_{q}\setminus\cup_{j=1}^{r}(b_{i_{j}}+V) $, then $TRS_{k,n}[\bm\alpha,1,k-1,\eta]$ is MDS.
\end{Theorem}
\begin{proof} It is similar to the proof of Theorem~\ref{the-2} by using Lemma~\ref{lem-2}. Some care must be taken that adding $\infty$ to a set of evaluation points preserves the MDS property. However, if a polynomial $f\in \mathcal{P}_{k,n}[1,k-1,\eta]$ satisfies $f(\infty)=0$, then its degree is at most $k-2$.
\end{proof}

Let the notions be as in Theorem~\ref{the-4}, then by Theorem~\ref{theorem-1} we have the following result.
\begin{Corollary}
If $|\mathbb{F}_{q}\setminus\cup_{j=1}^{r}(b_{i_{j}}+V)|>6$, then for any $n,k$ with $2<k<n-2$ and $n\leq | \cup_{j=1}^{r}(b_{i_{j}}+V)|$ there exists a non-GRS MDS $TRS_{k,n}[\bm\alpha,1,k-1,\eta]$ code.
\end{Corollary}
\begin{Remark}
The construction of Theorem~\ref{the-4} is $(+)$-twisted codes in\cite{P. Beelen} when $V=\{0\}$.
\end{Remark}
\begin{Remark}
From the result, the code length can reach $r|V|+1$. If the quotient group $\mathbb{F}_{q}/V$ has proper non-trivial subgroup, then our construction increases the code length. It is easy to see that $r|V|=p^{t}$ for some $t<m$. However, there is always a subgroup of $(\mathbb{F}_{q},+)$ with order $p^{t}$. So, Theorem~\ref{the-4} can not increase the maximal length which $(+)$-twisted codes reached.
\end{Remark}
\begin{Theorem}\label{the-5}
Let $\mathbb{F}_{p^{m}}$ be a finite field with $p$ an odd prime number and $m>1$. Let $V$ be a subgroup of $(\mathbb{F}_{q},+)$ with order $p^{m-1}$. Suppose $\mathbb{F}_{q}/V=\{\overline{0},\overline{b},\dots,(p-1)\overline{b}\}$ with $\overline{b}$ the generator of $\mathbb{F}_{q}/V$. Let $c_{1},c_{2},\dots,c_{p-2}\in\overline{b}=b+V$ be distinct. Let $(t,h)=(1,k-1), 1\leq k<n, \{\alpha_{1},\alpha_{2},...,\alpha_{n}\}$ be a subset of $V\bigcup\{c_{1},c_{2},\dots,c_{p-2},\infty\}$. If $(-\eta)^{-1}\in (p-1)\overline{b}=(p-1)b+V$, then $TRS_{k,n}[\bm\alpha,1,k-1,\eta]$ is MDS.
\end{Theorem}
\begin{proof} It is similar to the proofs of Theorem~\ref{the-3} and Theorem~\ref{the-4}.
\end{proof}

Let the notions be as in Theorem~\ref{the-5}, then by Theorem~\ref{theorem-1} we have the following result.
\begin{Corollary}
If $|V|>6$, then for any $n,k$ with $2<k<n-2$ and $n\leq |V\bigcup\{c_{1},c_{2},\dots,c_{p-2}\}|$ there exists a non-GRS MDS $TRS_{k,n}[\bm\alpha,1,k-1,\eta]$ code.
\end{Corollary}
\begin{Remark}
In the case $(t,h)=(1,k-1)$, from Theorem~\ref{the-5}, the code length can reach $p^{m-1}+ p-1$. However, the length of $(+)$-twisted codes can only reach $p^{m-1} + 1$. But for finite field $\mathbb{F}_{q}$ of even $q$, the minimal prime divisor of $q$ is $2$, then we can not increase length by this way.
\end{Remark}
\begin{Example}
Let $q=7^{2}$, $\mathbb{F}_{7^{2}}=\mathbb{F}_{7}(\theta)$ with $\theta^{2}+2=0$ and $V=(\mathbb{F}_{7},+)$. Let $\{b_{1},b_{2},b_{3},b_{4},b_{5}\}\subset\theta+V$, $\{\alpha_{1},\dots,\alpha_{13}\}=V\cup\{b_{1},b_{2},b_{3},b_{4},b_{5},\infty\}$. Then for any $1\leq k<13$ and $(-\eta)^{-1}\in 6\theta+V$, $TRS_{k,13}[\bm\alpha,1,k-1,\eta]$ are MDS. However, by the method in \cite{P. Beelen}, the length only reach $8$. Furthermore, there exists a non-GRS code.
\end{Example}

\begin{Remark}
Note that, for general $t$ and $h$, the length of MDS twisted Reed-Solomon codes over $\mathbb{F}_{q}$ can reach $\frac{q-1}{2}+1$ when $q$ is odd and $\frac{q}{2}+1$ when $q$ is even in \cite{P. Beelen}. If we only focus on the maximal length of MDS twisted Reed-Solomon codes, our constructions make no contribution.
\end{Remark}
\section{LCD MDS codes}
In this section, we give two constructions of LCD MDS Codes by using generalized twisted Reed-Solomon codes. The following two results are known.

\begin{Proposition}\label{pro-4.1}\cite{Mesnager S.}
Let $G$ be a generator matrix of a linear code $C$, then $C$ is self orthogonal if and only if $GG^{T}=0$.
\end{Proposition}
\begin{Lemma}\label{lem-4.1}\cite{Mesnager S.}
Let $C$ be an $[\,n, k, d\,]$ self orthogonal linear code generated by the matrix $G = [I_{k} : P]$. Then for any $\beta \in \mathbb{F}_{q}\backslash\{0, 1, -1\}$, the linear code $C_{\beta}$ generated by the matrix $G_{\beta} = [I_{k} : \beta P]$ is an $[\,n, k, d\,]$ LCD code.
\end{Lemma}

We have the following lemma.
\begin{Lemma}\label{lem-4.2}
Let $C$ be an $[\,n, k, n-k+1\,]$ self orthogonal MDS linear code generated by the matrix $G = [A_{k\times k} : B_{k\times(n-k)}]$. Then for any $\beta \in \mathbb{F}_{q}\backslash\{0, 1, -1\}$, the linear code $C_{\beta}$ generated by the matrix $G_{\beta} = [A_{k\times k} : \beta B_{k\times(n-k)}]$ is an $[\,n, k, n-k+1\,]$ LCD MDS code.
\end{Lemma}
\begin{proof} Since $C$ is an MDS code, then the matrix $A$ is nonsingular. So, $G^{'} = A^{-1}G = [I_{k} : A^{-1}B]$ is also a generator matrix of $C$. By Lemma~\ref{lem-4.1}, for any $\beta \in \mathbb{F}_{q}\backslash\{0, 1, -1\}$, the linear code $C^{'}_{\beta}$ generated by the matrix $G^{'}_{\beta} = [I_{k} : \beta A^{-1}B ]$ is an $[\,n, k, n-k+1\,]$ LCD MDS code. However, $C^{'}_{\beta}$ can be generated by $AG^{'}_{\beta} = [A : A(\beta A^{-1}B)] = [A : \beta B]$, which completes the proof.
\end{proof}

For any distinct elements $\alpha_{1},\alpha_{2},...,\alpha_{n}$ of $\mathbb{F}_{q}$, put $\bm\alpha=(\alpha_{1},\alpha_{2},...,\alpha_{n})$ and denote by $A_{\alpha}$ the matrix
$$
\left[
\begin{matrix}
1&1&\dots&1\\
\alpha_{1}&\alpha_{2}&\dots&\alpha_{n}\\
\vdots&\vdots&&\vdots\\
\alpha_{1}^{n-2}&\alpha_{2}^{n-2}&\dots&\alpha_{n}^{n-2}
\end{matrix}
\right]_{(n-1)\times n}.
$$

In~\cite{Jin L.}, the authors have shown that the solution space of the equation system $A_{\alpha}\bm X^{T}=\mathbf{0}$ has dimension $1$ and $\{\bm u=(u_{1},u_{2},...,u_{n})\}$ is a basis of this solution space, where $u_{i}=\prod_{1\leq j\leq n,j\neq i}(\alpha_{i}-\alpha_{j})^{-1}$, especially $u_{i}\neq 0$ for all $i\in\{1,...,n\}$. In the following of this section, for any distinct $\alpha_{1},\dots,\alpha_{n}$ of $\mathbb{F}_{q}$, we always let $\bm u=(u_{1},u_{2},...,u_{n})$ where $u_{i}=\prod_{1\leq j\leq n,j\neq i}(\alpha_{i}-\alpha_{j})^{-1}$.

\begin{Theorem}\label{the-4.1}
Let $\ell=1$, $t=1$, $1\leq k\leq\frac{n-2}{2}$, $0 \leq h \leq k-1 $, $\eta \in \mathbb{F}_{q}^{*}$, $\alpha_{1},\alpha_{2},...,\alpha_{n}\in \mathbb{F}_{q}$ be distinct, if $u_{1},u_{2},...,u_{n}$ are all nonzero square elements of $\mathbb{F}_{q}$ with $u_{i}=v_{i}^{2},i=1,2,...,n$. Let $\bm \alpha = (\alpha_{1},\dots,\alpha_{n})$, $\bm v = (v_{1},\dots,v_{n})$, then $GTRS_{k,n}[\bm\alpha,1,h,\eta,\bm v]$ is a self orthogonal code.
\end{Theorem}
\begin{proof}
By the condition above, a generator matrix $G$ of $GTRS_{k,n}[\bm\alpha,1,h,\eta,\bm v]$ is
$$
\left[
\begin{matrix}
v_{1}&v_{2}&\dots&v_{n}\\
v_{1}\alpha_{1}&v_{2}\alpha_{2}&\dots&v_{n}\alpha_{n}\\
\vdots&\vdots&&\vdots\\
v_{1}\alpha_{1}^{h-1}&v_{2}\alpha_{2}^{h-1}&\dots&v_{n}\alpha_{n}^{h-1}\\
v_{1}(\alpha_{1}^{h}+\eta\alpha_{1}^{k})&v_{2}(\alpha_{2}^{h}+\eta\alpha_{2}^{k})&\dots&v_{n}(\alpha_{n}^{h}+\eta\alpha_{n}^{k})\\
v_{1}\alpha_{1}^{h+1}&v_{2}\alpha_{2}^{h+1}&\dots&v_{n}\alpha_{n}^{h+1}\\
\vdots&\vdots&&\vdots\\
v_{1}\alpha_{1}^{k-1}&v_{2}\alpha_{2}^{k-1}&\dots&v_{n}\alpha_{n}^{k-1}
\end{matrix}
\right].
$$
Then the value of $(GG^{T})_{ij}$ has the following three cases:

(i) \,\,  $\sum_{i=1}^{n}u_{i}\alpha_{i}^{m_{1}}$.

(ii)  \,\, $\sum_{i=1}^{n}u_{i}\alpha_{i}^{h+m_{2}} + \eta\sum_{i=1}^{n}u_{i}\alpha_{i}^{k+m_{2}}$.

(iii) \,\, $\sum_{i=1}^{n}u_{i}\alpha_{i}^{2h} + 2\eta\sum_{i=1}^{n}u_{i}\alpha_{i}^{k+h} + \eta^{2}\sum_{i=1}^{n}u_{i}\alpha_{i}^{2k}$.

\noindent Since $k\leq\frac{n-2}{2},0\leq h\leq k-1$, then $m_{1},h+m_{2},k+m_{2},2h,2k,k+h$ are all less than $n-1$.
By $A_{\bm\alpha}\bm u^{T}= \bm 0$, we get $GG^{T}=0$.
\end{proof}
\begin{Corollary}\label{cor-4.1}
Let $\mathbb{F}_{2^{m}}$ be a finite field with $m\geq2$, $V$ be a proper subgroup of $(\mathbb{F}_{2^{m}},+)$ with order $s$, $L=\{\overline{b_{1}},\dots,\overline{b_{\frac{2^{m-1}}{s}}}\}$ be a proper subgroup of quotient group $\mathbb{F}_{2^{m}}/V$ with order $\frac{2^{m-1}}{s}$. Suppose $\ell=1$, $(t,h)=(1,k-1), 1\leq k\leq\frac{n-2}{2}, \{\alpha_{1},\alpha_{2},...,\alpha_{n}\}$ is a subset of $\cup_{j=1}^{\frac{2^{m-1}}{s}}(b_{j}+V)$, and $\eta^{-1}\in \mathbb{F}_{2^{m}}\setminus \cup_{j=1}^{\frac{2^{m-1}}{s}}(b_{j}+V)$. Let $\beta \in \mathbb{F}_{q}\backslash\{0, 1, -1\}$, $u_{i}=v_{i}^{2},i=1,2,...,n$, and let $\bm \alpha = (\alpha_{1},\dots,\alpha_{n})$, $\bm v^{'}= (v_{1},\dots,v_{k},\beta v_{k+1},\dots,\beta v_{n})$, then $GTRS_{k,n}[\bm\alpha,1,k-1,\eta,\bm v^{'}]$ is an LCD MDS code.

\end{Corollary}
\begin{proof} Since any element of $\mathbb{F}_{2^{m}}$ is square element, then $v_{i}(i=1,\dots,n)$ are existing. Let $\bm v= (v_{1},\dots,v_{n})$, by Theorem ~\ref{the-4}, $GTRS_{k,n}[\bm\alpha,1,k-1,\eta,\bm v]$ is an MDS code. From Theorem~\ref{the-4.1}, $GTRS_{k,n}[\bm\alpha,1,k-1,\eta,\bm v]$ is a self orthogonal code. Finally, by Lemma~\ref{lem-4.2}, $GTRS_{k,n}[\bm\alpha,1,k-1,\eta,\bm v^{'}]$ is an LCD MDS code.
\end{proof}
\begin{Remark}
We choose $\{\alpha_{1},\alpha_{2},...,\alpha_{n}\}=L$ as in Corollary~\ref{cor-4.1}, then $n=2^{m-1}$. For any $1\leq k\leq\frac{n-2}{2}$, we give a construction of $[\,n,k,n-k+1\,]$ LCD MDS codes.
\end{Remark}
\begin{Remark}
If $2<k<n-2$, $\mid \mathbb{F}_{2^{m}}\setminus L\mid>6$ i.e. $m\geq4$, then our construction contains many non-GRS LCD MDS codes.
\end{Remark}
The results above are about even characteristic. Next, we give LCD MDS codes in odd characteristic. Firstly, we give some useful results.
\begin{Proposition}\cite{Mesnager S.}
If $G$ is a generator matrix of an $[n,k]$ linear code $C$, then $C$ is an LCD code if and only if $GG^{T}$ is nonsingular.
\end{Proposition}
\begin{Proposition}\cite{Mesnager S.}
Let $G=\{\alpha_{1},\alpha_{2},\dots,\alpha_{n}\}$ be a subgroup of $\mathbb{F}_{q}^{*}$ with order $n$, then for any integer $t$, we have

\begin{equation}\label{eq-1}
\alpha_{1}^{t}+\alpha_{2}^{t}+\cdots+\alpha_{n}^{t}=
\begin{cases}
n  & t\equiv 0\pmod{n}, \\
0  & \text otherwise.
\end{cases}
\end{equation}

\end{Proposition}
\begin{Theorem}\label{the-4.2}\cite{M. Bossert}
Let $q_{0}$ be a prime power, and $1=s_{0}<\dots<s_{\ell}$ be non-negative integers such that $\mathbb{F}_{q_{0}^{s_{0}}}\subset\mathbb{F}_{q_{0}^{s_{1}}}\subset\dots\subset\mathbb{F}_{q_{0}^{s_{\ell}}}=\mathbb{F}_{q}$
is a chain of subfields. Fix $k<n\leq q_{0}$ and the entries of $\bm\alpha=(\alpha_{1},\dots,\alpha_{n})\in\mathbb{F}_{q_{0}}^{n}$ as pairwise distinct locators. Finally, let $\bm t$, $\bm h$ and $\bm \eta$ be chosen as in Definition~\ref{def-2}, such that $\eta_{i}\in\mathbb{F}_{q_{0}^{s_{i}}}\backslash\mathbb{F}_{q_{0}^{s_{i-1}}}$ for $i=1,\dots,\ell$. Then $TRS_{k,n}[\bm\alpha,\mathbf{t},\mathbf{h},\bm\eta]$ is MDS.
\end{Theorem}
\begin{Theorem}\label{the-4.3}
Let $q=p^{m},~\{\alpha_{1},\alpha_{2},...,\alpha_{n}\}$ a subgroup of $(\mathbb{F}_{q}^{*},~\cdot)$, $\bm t=(1,2,...,k),~ \bm h=(0,1,...,k-1),~\bm \eta=(\eta_{1},\eta_{2},...,\eta_{k})\in(\mathbb{F}_{q}^{*})^{n},~k\leq\frac{n}{2}$. If $k,~n$ satisfy one of the following conditions, then $TRS_{k,n}[\bm\alpha,\mathbf{t},\mathbf{h},\bm\eta]$ is LCD.

(A) $k\geq3, k=\frac{n+1}{3}$.

(B) $k\geq3, k=\frac{n}{3}$.

(C) $k\geq3, \frac{n+2}{3}\leq k<\frac{n}{2}$, and all entries of $(\eta_{n-2k+2},\eta_{n-2k+3},...,\eta_{k})+(\eta_{k},\eta_{k-1},...,\eta_{n-2k+2})$ are nonzero.

(D) $k\geq3, k=\frac{n}{2}$, $1+\eta_{1}^{2}$ is nonzero and all entries of $(\eta_{2},\eta_{3},...,\eta_{k})+(\eta_{k},\eta_{k-1},...,\eta_{2})$ are nonzero.

(E) $k=2, n=4$, and $2(1+\eta_{1}^{2})\eta_{2}$ is nonzero.

(F) $k=2, n=5$.

(G) $k=2, n=6$.

(H) $k=1, n=2$, and $1+\eta_{1}^{2}$ is nonzero.

(I) $k=1, n>2$.

\end{Theorem}
\begin{proof}
Since $\{\alpha_{1},\alpha_{2},...,\alpha_{n}\}$ is a subgroup of $\mathbb{F}_{q}^{*}$, then $\gcd(n,p)=1$.
Let
$$
A=\left[
\begin{matrix}
1&1&\dots&1\\
\alpha_{1}&\alpha_{2}&\dots&\alpha_{n}\\
\vdots&\vdots&&\vdots\\
\alpha_{1}^{k-1}&\alpha_{2}^{k-1}&\dots&\alpha_{n}^{k-1}
\end{matrix}
\right]_{k\times n},
$$

$$
B=\left[
\begin{matrix}
\eta_{1}\alpha_{1}^{k}&\eta_{1}\alpha_{2}^{k}&\dots&\eta_{1}\alpha_{n}^{k}\\
\eta_{2}\alpha_{1}^{k+1}&\eta_{2}\alpha_{2}^{k+1}&\dots&\eta_{2}\alpha_{n}^{k+1}\\
\vdots&\vdots&&\vdots\\
\eta_{k}\alpha_{1}^{2k-1}&\eta_{k}\alpha_{2}^{2k-1}&\dots&\eta_{k}\alpha_{n}^{2k-1}
\end{matrix}
\right]_{k\times n}.
$$
From the condition, a generator matrix of $TRS_{k,n}[\bm\alpha,\mathbf{t},\mathbf{h},\bm\eta]$ is $G=A+B$.
We have $GG^{T}=(A+B)(A+B)^{T}=AA^{T}+AB^{T}+BA^{T}+BB^{T}$. By Equation (4.1),
$$
AA^{T}=\left[
\begin{matrix}
n&0&\dots&0\\
0&0&\dots&0\\
\vdots&\vdots&&\vdots\\
0&0&\dots&0
\end{matrix}
\right]_{k\times k}.
$$
When (A) holds, we have $AB^{T}=0$, $BA^{T}=0$ and
$$
BB^{T}=\left[
\begin{matrix}
0&\dots&0&\eta_{1}\eta_{k}n\\
0&\dots&\eta_{2}\eta_{k-1}n&0\\
\vdots&&\vdots&\vdots\\
\eta_{k}\eta_{1}n&\dots&0&0
\end{matrix}
\right]_{k\times k}.
$$
Then
$$
GG^{T}=\left[
\begin{matrix}
n&\dots&0&\eta_{1}\eta_{k}n\\
0&\dots&\eta_{2}\eta_{k-1}n&0\\
\vdots&&\vdots&\vdots\\
\eta_{k}\eta_{1}n&\dots&0&0
\end{matrix}
\right]_{k\times k}
$$
is nonsingular.

When (B) holds, we have $AB^{T}=0$, $BA^{T}=0$ and
$$
BB^{T}=\left[
\begin{matrix}
0&\dots&0&0\\
0&\dots&0&\eta_{2}\eta_{k}n\\
0&\dots&\eta_{3}\eta_{k-1}n&0\\
\vdots&&\vdots&\vdots\\
0&\eta_{k}\eta_{2}n&\dots&0
\end{matrix}
\right]_{k\times k}.
$$
Then
$$
GG^{T}=\left[
\begin{matrix}
n&\dots&0&0\\
0&\dots&0&\eta_{2}\eta_{k}n\\
0&\dots&\eta_{3}\eta_{k-1}n&0\\
\vdots&&\vdots&\vdots\\
0&\eta_{k}\eta_{2}n&\dots&0
\end{matrix}
\right]_{k\times k}
$$
is nonsingular.

When (C) holds, $AB^{T}=(BA^{T})^{T}$, we have
$$
BA^{T}=\left[
\begin{matrix}
0&\dots&0&\dots&\dots&0\\
0&\dots&0&\dots&\dots&0\\
\vdots&&\vdots&&&\vdots\\
0&\dots&0&\dots&\dots&0\\
0&\dots&0&\dots&0&\eta_{n-2k+2}n\\
0&\dots&0&\dots&\eta_{n-2k+3}n&0\\
\vdots&&\vdots&&\vdots&\vdots\\
0&\dots&0&\eta_{k}n&\dots&0
\end{matrix}
\right]_{k\times k},
$$
$$
BB^{T}=\left[
\begin{matrix}
0&\dots&0&\eta_{1}\eta_{n-2k+1}n&0&\dots&0\\
0&\dots&\eta_{2}\eta_{n-2k}n&0&0&\dots&0\\
\vdots&\vdots&\vdots&\vdots&\vdots&\vdots&\vdots\\
\eta_{n-2k+1}\eta_{1}n&0&\dots&\dots&\dots&\dots&0\\
0&0&\dots&\dots&\dots&\dots&0\\
\vdots&&&&&&\vdots\\
0&0&\dots&\dots&\dots&\dots&0
\end{matrix}
\right]_{k\times k}.
$$
Then
$$
GG^{T}=\left[
\begin{matrix}\begin{smallmatrix}
n&\dots&0&\eta_{1}\eta_{n-2k+1}n&0&\dots&0\\
0&\dots&\eta_{2}\eta_{n-2k}n&0&0&\dots&0\\
\vdots&\vdots&\vdots&\vdots&\vdots&\vdots&\vdots\\
\eta_{n-2k+1}\eta_{1}n&0&\dots&\dots&\dots&\dots&0\\
0&\dots&\dots&0&\dots&0&(\eta_{n-2k+2}+\eta_{k})n\\
0&\dots&\dots&0&\dots&(\eta_{n-2k+3}+\eta_{k-1})n&0\\
\vdots&&&\vdots&&\vdots&\vdots\\
0&\dots&\dots&0&(\eta_{k}+\eta_{n-2k+2})n&\dots&0
\end{smallmatrix}\end{matrix}
\right]_{k\times k}
$$
is nonsingular.

When (D) holds, we have
$$
BA^{T}=\left[
\begin{matrix}
0&\dots&0&0\\
0&\dots&0&\eta_{2}n\\
0&\dots&\eta_{3}n&0\\
\vdots&&\vdots&\vdots\\
0&\eta_{k}n&\dots&0
\end{matrix}
\right]_{k\times k},
$$
$$
BB^{T}=\left[
\begin{matrix}
\eta_{1}^{2}n&0&\dots&0\\
\vdots&&\vdots&\vdots\\
0&0&\dots&0
\end{matrix}
\right]_{k\times k}.
$$
Then
$$
GG^{T}=\left[
\begin{matrix}
(\eta_{1}^{2}+1)n&\dots&0&0\\
0&\dots&0&(\eta_{2}+\eta_{k})n\\
0&\dots&(\eta_{3}+\eta_{k-1})n&0\\
\vdots&&\vdots&\vdots\\
0&(\eta_{k}+\eta_{2})n&\dots&0
\end{matrix}
\right]_{k\times k}
$$
is nonsingular.

The statements (E), (F), (G), (H), (I) are calculated similarly.
\end{proof}
\begin{Corollary}\label{cor-4.2}
Let $q$ be an odd prime power, and $1=s_{0}<\dots<s_{\ell}$ be non-negative integers such that $\mathbb{F}_{q_{0}}=\mathbb{F}_{q_{0}^{s_{0}}}\subsetneq\mathbb{F}_{q_{0}^{s_{1}}}\subsetneq\dots\subsetneq\mathbb{F}_{q_{0}^{s_{k}}}=\mathbb{F}_{q}$
is a chain of subfields. Let $\{\alpha_{1},\alpha_{2},...,\alpha_{n}\}$ be a subgroup of $\mathbb{F}_{q_{0}}^{*}$, and let $\eta_{i}\in\mathbb{F}_{q_{0}^{s_{i}}}\backslash\mathbb{F}_{q_{0}^{s_{i-1}}}$ for $i=1,\dots,k$ with $1+\eta_{1}^{2}\neq0$, $\eta_{j}\neq 0$ for $i=2,\dots,k$. Suppose $\bm t=(1,2,...,k), \bm h=(0,1,...,k-1),\bm \eta=(\eta_{1},\eta_{2},...,\eta_{k})\in(\mathbb{F}_{q}^{*})^{n}$. Let $k\leq\frac{n}{2}$, if $k$, $n$ satisfy one of the following conditions,

(A) $k\geq3, k=\frac{n+1}{3}$;

(B) $k\geq3, k=\frac{n}{3}$;

(C) $k\geq3, \frac{n+2}{3}\leq k<\frac{n}{2}$;

(D) $k\geq3, k=\frac{n}{2}$;

(E) $k=1, n>2$;

\noindent then $TRS_{k,n}[\bm\alpha,\mathbf{t},\mathbf{h},\bm\eta]$ is an LCD MDS code.
\end{Corollary}
\begin{proof} We simply apply Theorem~\ref{the-4.2} and Theorem~\ref{the-4.3}. Some care must be taken in condition (C) and (D). For condition (C), if the length of $(\eta_{n-2k+2},\eta_{n-2k+3},\dots,\eta_{k})$ is even, then all entries of $(\eta_{n-2k+2},\eta_{n-2k+3},...,\eta_{k})+(\eta_{k},\eta_{k-1},...,\eta_{n-2k+2})$ are $\eta_{s}+\eta_{j}$ with $s\neq j$. From the choice of $\bm \eta$, we know $\eta_{s}$ and $\eta_{j}$ are not in same field with different indices. Then, all entries of $(\eta_{n-2k+2},\eta_{n-2k+3},...,\eta_{k})+(\eta_{k},\eta_{k-1},...,\eta_{n-2k+2})$ are nonzero. If the length of $(\eta_{n-2k+2},\eta_{n-2k+3},\dots,\eta_{k})$ is odd, then there is one position with same index and others are not same. However, all the fields are odd characteristic and $\eta_{j}\neq 0$ for $i=2,\dots,k$. So all entries of $(\eta_{n-2k+2},\eta_{n-2k+3},...,\eta_{k})+(\eta_{k},\eta_{k-1},...,\eta_{n-2k+2})$ are nonzero. It is similar to explain condition (D).
\end{proof}
\begin{Remark}
From the state of Corollary~\ref{cor-4.2}, we know $n|(q_{0}-1)$ for some $q_{0}$ an odd prime power. Then, for the case $k=2$ in Theorem~\ref{the-4.3}, we need some special $q_{0}$. For a sufficient large $q$, we can construct LCD MDS codes when $3\leq k$ and $\frac{n}{3}\leq k\leq \frac{n}{2}$.
\end{Remark}
\section{Conclusion}
This paper presents new constructions of MDS twisted Reed-Solomon codes and in some cases, our constructions can get some MDS codes with the length longer than the constructions of previous works. And we also give two new constructions of LCD MDS codes from generalized twisted Reed-Solomon codes. However, the code length is also restrained for general $t$ and $h$. It will be significant to construct an MDS twisted Reed-Solomon code with a longer length.

\vskip 4mm

\noindent {\bf Acknowledgement.} This work was supported by NSFC (Grant No. 11871025).

%\bibliography{mybibfile}

\begin{thebibliography}{1}
\bibitem{M. Bossert}
P. Beelen, M. Bossert, S. Puchinger, J. Rosenkilde n\'{e} Nielsen, Structural properties of twisted Reed-Solomon codes with applications to Code-Based Cryptography, IEEE ISIT. 946-950 (2018).

\bibitem{P. Beelen2}
P. Beelen, L. Jin, Explicit MDS codes with complementary duals, IEEE Trans. Inform. Theory. 64(11), 7188-7193 (2018).

\bibitem{P. Beelen}
P. Beelen, S. Puchinger, J. Rosenkilde n\'{e} Nielsen, Twisted Reed-Solomon codes, IEEE ISIT. 336-340 (2017).

\bibitem{Guilley}
C. Carlet, S. Guilley, Complementary dual codes for counter-measures to side-channel attacks, J. Adv. Math. Commun. 10(1), 131-150 (2016).

\bibitem{Mesnager S.}
C. Carlet, S. Mesnager, C. Tang, Y. Qi, Euclidean and Hermitian LCD MDS codes, Des. Codes Cryptogr.
86(11), 2605-2618 (2018).

\bibitem{Carlet C.}
C. Carlet, S. Mesnager, C. Tang, Y. Qi, R. Pellikaan, Linear codes over $\mathbb{F}_{q}$ are equivalent to LCD codes for $q>3$, IEEE Trans. Inform. Theory. 64(4), 3010-3017 (2018).

\bibitem{B. Chen}
B. Chen, H. Liu, New constructions of MDS codes with complementary duals, IEEE Trans. Inform. Theory. 64(8), 5776-5782 (2018).

\bibitem{H. Q. Dinh}
H. Q. Dinh, T. Bag, A. K. Upadhyay, R. Bandi, W. Chinnakum, On the structure of cyclic codes over $\mathbb{F}_{q}{RS}$ and applications in quantum and LCD Codes constructions, IEEE Access. 8, 18902-18914 (2020).

\bibitem{S. T. Dougherty}
S. T. Dougherty, J.-L. Kim, B. Ozkaya, L. Sok, P. Sol\'{e}, The combinatorics of LCD codes: Linear programming bound and orthogonal matrices, International Journal of Information and Coding Theory. 4(2-3), 116-128 (2017).

\bibitem{M. Esmaeili}
M. Esmaeili, S. Yari, On complementary-dual quasi-cyclic codes, Finite Fields Appl. 15(3), 375-386 (2009).

\bibitem{C. Galindo}
C. Galindo, O. Geil, F. Hernando, D. Ruano, New binary and ternary LCD codes, IEEE Trans. Inform. Theory. 65(2), 1008-1016 (2019).

\bibitem{C. Guneri}
C. G\"{u}neri, B. \"{O}zkaya, P. Sol\'{e}, Quasi-cyclic complementary dual codes, Finite Fields Appl. 42, 67-80 (2016).

\bibitem{X. Hou}
X. Hou, F. Oggier, On LCD codes and lattices, IEEE ISIT. 1501-1505 (2016).

\bibitem{Jin L.2}
L. Jin, Construction of MDS codes with complementary duals, IEEE Trans. Inform. Theory. 63(5), 2843-2847 (2017).

\bibitem{Jin L.}
L. Jin, C. Xing, New MDS self-dual codes from generalized Reed-Solomon codes, IEEE Trans. Inform. Theory. 63(3), 1434-1438 (2017).

\bibitem{Julien}
J. Lavauzelle, J. Renner, Cryptanalysis of a system based on twisted Reed-Solomon codes, Des. Codes Cryptogr. 88(7), 1285-1300 (2020).

\bibitem{C. Ding}
C. Li, C. Ding, H. Liu, Parameters of two classes of LCD BCH codes, arXiv:1608.02670 (2016).

\bibitem{S. Li}
S. Li, C. Ding, H. Liu, A family of reversible BCH codes, arXiv:1608.02169 (2016).

\bibitem{C. Li}
C. Li, C. Ding, S. Li, LCD cyclic codes over finite fields, IEEE Trans. Inform. Theory. 63(7), 4344-4356 (2017).

\bibitem{Fengwei Li}
F. Li, Q. Yue, Y. Wu, LCD and self-orthogonal group codes in a finite abelian $p$-group algebra, IEEE Trans. Inform. Theory. 66(5), 2717-2728 (2020).

\bibitem{Hongwei Liu}
H. Liu, X. Pan, Galois hulls of linear codes over finite fields, Des. Codes Cryptogr. 88(2), 241-255 (2020).

\bibitem{J. L. Massey}
J. L. Massey, Linear codes with complementary duals, Discrete Math. 106-107, 337-342 (1992).

\bibitem{J. Qian}
J. Qian, L. Zhang, On MDS linear complementary dual codes and entanglement-assisted quantum codes, Des. Codes Cryptogr. 86(7), 1565-1572 (2018).

\bibitem{R. Roth}
R. Roth, Introduction to Coding Theory, Cambridge, UK: Cambridge University Press, (2006).

\bibitem{Xueying Shi}
X. Shi, Q. Yue, S. Yang, New LCD MDS codes constructed from generalized Reed-Solomon codes, J. Alg. Appl. 1950150 (2018).

\bibitem{Y. Wu}
Y. Wu, Y. Lee, Binary LCD codes and self-orthogonal codes via simplicial complexes, IEEE Trans. Inform. Theory. 24(6), 1159-1162(2020).

\bibitem{Y. Wu2}
Y. Wu, Q. Yue, Factorizations of binomial polynomials and enumerations of LCD and self-dual constacyclic codes, IEEE Trans. Inform. Theory. 65(3), 1740-1751 (2019).

\bibitem{X. Yang}
X. Yang, J. L. Massey, The necessary and sufficient condition for a cyclic code to have a complementary dual, Discrete Math. 126(1-3), 391-393 (1994).

\bibitem{Z. Zhou}
Z. Zhou, X. Li, C. Tang, C. Ding, Binary LCD codes and self-orthogonal codes from a generic construction, IEEE Trans. Inform. Theory. 65(1), 16-27 (2019).


\end{thebibliography}

\end{document}